\documentclass[journal]{article}

\usepackage{graphicx,url}
\usepackage{dcolumn}
\usepackage{bm}
\usepackage{amsmath,amsfonts,amssymb}

\newtheorem{theorem}{\noindent{\it Theorem}}[section]

\newtheorem{remark}[theorem]{\noindent{\it Remark}}

\newenvironment{proof}{\noindent{\it Proof:}}{$\hfill$ $\Box$\\ }
\newtheorem{example}{\noindent{\it Example}}[section]

\begin{document}

\title{Asymmetric quantum convolutional codes}

 \author{Giuliano G. La Guardia
\thanks{Giuliano Gadioli La Guardia is with Department of Mathematics and Statistics,
State University of Ponta Grossa (UEPG), 84030-900, Ponta Grossa,
PR, Brazil. }}

\maketitle
\begin{abstract}
In this paper, we construct the first families of asymmetric quantum
convolutional codes (AQCC)'s. These new AQCC's are constructed by
means of the CSS-type construction applied to suitable families of
classical convolutional codes, which are also constructed here. The
new codes have noncatastrophic generator matrices and they have
great asymmetry. Since our constructions are performed
algebraically, i.e., we develop general algebraic methods and
properties to perform the constructions, it is possible to derive
several families of such codes and not only codes with specific
parameters. Additionally, several different types of such codes are
obtained.
\end{abstract}


\section{Introduction}\label{Intro}

Several works available in literature deal with constructions of
quantum error-correcting codes (QECC, for short) and asymmetric
quantum error-correcting codes (AQECC)
\cite{Calderbank:1998,Nielsen:2000,Ashikmin:2001,Ketkar:2006,Ioffe:2007,LaGuardia:2009,Sarvepalli:2009,LaGuardia:2011,LaGuardia:2012,LaGuardia:2013}.
In contrast with this subject of research one has the theory of
quantum convolutional codes
\cite{Ollivier:2003,Ollivier:2004,Aido:2004,Grassl:2007,Aly:2007,Klapp:2007,Forney:2007,LaGuardia:2014,LaGuardia:2014A}.
Ollivier and Tillich \cite{Ollivier:2003,Ollivier:2004} were the
first to develop the stabilizer structure for these codes. Almeida
and Palazzo Jr. constructed an $[(4, 1, 3)]$ (memory $\mu=3$)
quantum convolutional code \cite{Aido:2004}. Grassl and R\"otteler
\cite{Grassl:2006,Grassl:2007} generated quantum convolutional codes
as well as they provide algorithms to obtain non-catastrophic
encoders. Forney \emph{et al.} constructed rate $(n-2)/n$ quantum
convolutional codes.

An asymmetric quantum convolutional code (AQCC) is a quantum code
defined over quantum channels where qudit-flip errors and
phase-shift errors may have different probabilities. As it is well
known, Steane \cite{Steane:1996} was the first who introduced the
notion of asymmetric quantum errors. The parameters of an AQCC will
be denoted by $[(n, k, \mu; \gamma,
{[d_{z}]}_{f}/{[d_{x}]}_{f})]_{q}$, where $n$ is the frame size, $k$
is the number of logical qudits per frame, $\mu$ is the memory,
${[d_{z}]}_{f}$ (${[d_{x}]}_{f}$) is the free distance corresponding
to phase-shift (qudit-flip) errors and $\gamma$ is the degree of the
code. The combined amplitude damping and dephasing channel (see
\cite{Sarvepalli:2009} the references therein) is a quantum channel
whose probability of occurrence of phase-shift errors is greater
than the probability of occurrence of qudit-flip errors.

In this paper we propose constructions of the first families of
asymmetric quantum convolutional codes. The constructions presented
here are performed algebraically (as mentioned above).

The first families of AQCC's presented in this paper, i.e.,
Construction I, are obtained from the construction method proposed
in Subsection~\ref{subIIIA}. This construction method is general,
i.e., it holds for every choice of sets of $m < n$ linearly
independent vectors ${\bf v}_{i}\in { \mathbb F}_{q}^{n}$, $i=1, 2,
\ldots, m$ (see the proof of Theorem~\ref{main1}). The AQCC's
derived from Construction I have parameters

\begin{itemize}

\item $[(n, \operatorname{rk}H_0, {\mu}^{*}; {\gamma}_{1}+{\gamma}_{2},
{(d_{z})}_{f}/{(d_{x})}_{f})]_{q},$ \\ where ${\gamma}_{1}=
\mu(\operatorname{rk}H_{\mu}+\operatorname{rk}H_{\mu}^{'})+\displaystyle\sum_{i=1}^{\mu-1}(\mu
-i)[\operatorname{rk}H_{(\mu-i)}^{'}-\operatorname{rk}H_{(\mu-i+1)}^{'}]$,
${\gamma}_{2}=
\mu(\operatorname{rk}H_{\mu}^{'})+\displaystyle\sum_{i=1}^{\mu-1}(\mu
-i)[\operatorname{rk}H_{(\mu-i)}^{'}-\operatorname{rk}H_{(\mu-i+1)}^{'}]$,
${(d_{x})}_{f}\geq{(d_1)}_{f} \geq d^{\perp}$ and ${(d_{z})}_{f}\geq
{(d_{2})}_f^{\perp} $, where ${(d_1)}_{f}, d^{\perp},
{(d_{2})}_f^{\perp}$ and the matrices $H_{0}, H_{0}^{'}, H_1,
H_{1}^{'}, \ldots,$ $\ldots, H_{\mu}, H_{\mu}^{'}$ are defined in
the proof of Theorem~\ref{main1}.

\end{itemize}

In Construction II (see Subsection~\ref{subIIIB}) we present
families of AQCC's derived from classical maximum-distance-separable
BCH codes:

\begin{itemize}

\item $[(n, 2i-4, {\mu}^{*}; 6, {[d_{z}]}_{f}/{[d_{x}]}_{f})]_{q}$,\\ where
$q=2^{t}$, $t\geq 4$, $n=q+1$, ${(d_{z})}_{f}\geq n-2i-1$ and
${(d_{x})}_{f}\geq 3$, for all $3\leq i \leq \frac{q}{2} - 1$;

\item $[(n, 2i-2t-2, {\mu}^{*}; 6,
{[d_{z}]}_{f}\geq n-2i-1/{[d_{x}]}_{f}\geq 2t+3)]_{q}$,\\ where
$q=2^{l}$, $l\geq 4$, $n=q+1$, $t$ integer with $1\leq t\leq i-2$,
$3\leq i \leq \frac{q}{2}$;

\item $[(n, 2i-2t, {\mu}^{*}; 4, {[d_{z}]}_{f}/{[d_{x}]}_{f})]_{q}$,\\ where
${(d_{z})}_{f}\geq n-2i-1$ and ${(d_{x})}_{f}\geq 2t+3$, $q=2^{l}$,
$l\geq 4$, $n=q+1$, $t$ integer with $1\leq t\leq i-1$, $2\leq i
\leq \frac{q}{2}$;

\item $[(n, 2i-2t-2, {\mu}^{*}; 6,
{[d_{z}]}_{f}/{[d_{x}]}_{f})]_{q}$,\\ where $q=p^{l}$, $p$ is an odd
prime, $l\geq 2$, $n=q+1$, ${(d_{z})}_{f}\geq n-2i$ and
${(d_{x})}_{f}\geq 2t+2$, for all $1\leq t\leq i-2$, where $3\leq i
\leq \frac{n}{2} - 1$;

\item $[(n, 2i-2t, {\mu}^{*}; 4,
{[d_{z}]}_{f}/{[d_{x}]}_{f})]_{q}$,\\
where $q=p^{l}$, $p$ is an odd prime, $l\geq 2$, $n=q+1$,
${(d_{z})}_{f}\geq n-2i$ and ${(d_{x})}_{f}\geq 2t+2$, for all
$1\leq t\leq i-1$, with $2\leq i \leq \frac{n}{2} - 1$;

\end{itemize}

In Construction III (see Subsection~\ref{subIIIC}) we construct
families of AQCC's derived from classical Reed-Solomon and
generalized Reed-Solomon codes. The AQCC's shown in Construction II
are distinct of the AQCC's shown in Construction III.

\begin{itemize}

\item $[(q-1, i-t-1, {\mu}^{*}; 3,
{[d_{z}]}_{f}/{[d_{x}]}_{f})]_{q}$,\\ where $q \geq 8$ is a prime
power ${(d_{z})}_{f}\geq q-i-1$ and ${(d_{x})}_{f}\geq t+2$, for all
$1\leq t\leq i-2$, where $3\leq i \leq q-3$;

\item $[(q-1, i-t, {\mu}^{*}; 2, {[d_{z}]}_{f}/{[d_{x}]}_{f})]_{q}$,\\ where
${(d_{z})}_{f}\geq q-i-1$, ${(d_{x})}_{f}\geq t+2$, for all $1\leq
t\leq i-1$, where $2\leq i \leq q-3$;

\item $[(n, n-t-k-2, {\mu}^{*}; 3, {[d_{z}]}_{f}/{[d_{x}]}_{f})]_{q}$,\\
where ${(d_{z})}_{f}\geq t+2$ and ${(d_{x})}_{f}\geq k+1$, where $q
\geq 5$ is a prime power, $k\geq 1$ and $n$ are integers such that
$5\leq n\leq q$ and $k\leq n-4$ and $t$ is an integer with $1\leq
t\leq n-k-2$;

\item $[(n, n-t-k-1, {\mu}^{*}; 2,
{[d_{z}]}_{f}/{[d_{x}]}_{f})]_{q}$,\\ where $q\geq 5$ is a prime
power, $k \geq 1$, $n \geq 5$ are integers such that $n\leq q$,
$k\leq n-4$, $1\leq t\leq n-k-1$, ${(d_{z})}_{f}\geq t+2$ and
${(d_{x})}_{f}\geq k+1$.
\end{itemize}

The ideas utilized in Constructions II and III are very similar to
that shown in Construction I, although in the latter constructions
it is possible to compute precisely the parameters are of the AQCC's
due to the structure of the classical BCH codes and (generalized)
Reed-Solomon codes involved in the construction process.

The paper is arranged as follows. In Section~\ref{II}, we recall the
concepts of convolutional and quantum convolutional codes. In
Section~\ref{III}, we present the contributions of this work,
\emph{i.e.}, the first families of asymmetric quantum convolutional
codes are constructed. In Section~\ref{IV}, we discuss the results
presented in this paper and, in Section~\ref{V}, the final remarks
are drawn.

\section{Background}\label{II}

\emph{Notation.} Throughout this paper, $p$ denotes a prime number,
$q$ is a prime power and ${\mathbb F}_{q}$ is a finite field with
$q$ elements. The code length is denoted by $n$ and we always assume
that $\gcd (q, n) = 1$. As usual, the multiplicative order of $q$
modulo $n$ is denoted by $l= {\operatorname{ord}}_{n}(q)$, and
$\alpha$ is considered a primitive $n$-th root of unity in the
extension field ${\mathbb F}_{{q}^{l}}$. The parameters of a linear
block code over ${\mathbb F}_{q}$, of length $n$, dimension $k$ and
minimum distance $d$, is denoted by $[n, k, d{]}_{q}$. Sometimes, we
abuse the notation by writing $C=[n, k, d{]}_{q}$. If $C$ is a
linear code then ${C}^{{\perp}}$ denotes its Euclidean dual.

\subsection{Review of Convolutional Codes}\label{IIA}

Convolutional codes are extensively investigated in the literature
\cite{Forney:1970,Lee:1976,Piret:1988,Rosenthal:1999,Johannesson:1999,Huffman:2003,Rosenthal:2001}.
Recall that a polynomial encoder matrix $G(D) \in {\mathbb
F}_{q}{[D]}^{k \times n}$ is called \emph{basic} if $G(D)$ has a
polynomial right inverse. A basic generator matrix is called
\emph{reduced} (or minimal, see \cite{Rosenthal:2001,Huffman:2003})
if the overall constraint length $\gamma
=\displaystyle\sum_{i=1}^{k} {\gamma}_i$, where ${\gamma}_i =
{\max}_{1\leq j \leq n} \{ \deg g_{ij} \}$, has the smallest value
among all basic generator matrices. In this case, we say that
$\gamma$ is the \emph{degree} of the resulting code.

A rate $k/n$ \emph{convolutional code} $C$ with parameters $(n, k,
\gamma ; \mu, d_{f} {)}_{q}$ is a submodule of ${\mathbb F}_q
{[D]}^{n}$ generated by a reduced basic matrix $G(D)=(g_{ij}) \in
{\mathbb F}_q {[D]}^{k \times n}$, \emph{i.e.}, $C = \{ {\bf
u}(D)G(D) | {\bf u}(D)\in {\mathbb F}_{q} {[D]}^{k} \}$, where $n$
is the length, $k$ is the dimension, $\gamma
=\displaystyle\sum_{i=1}^{k} {\gamma}_i$ is the degree, $\mu =
{\max}_{1\leq i\leq k}\{{\gamma}_i\}$ is the memory and
$d_{f}=$wt$(C)=\min \{wt({\bf v}(D)) \mid {\bf v}(D) \in C, {\bf
v}(D)\neq 0 \}$ is the free distance of the code. In the above
definition, the \emph{weight} of an element ${\bf v}(D)\in {\mathbb
F}_{q} {[D]}^{n}$ is defined as wt$({\bf
v}(D))=\displaystyle\sum_{i=1}^{n}$wt$(v_i(D))$, where wt$(v_i(D))$
is the number of nonzero coefficients of $v_{i}(D)$. In the field of
Laurent series ${\mathbb F}_{q}((D))$, whose elements are given by
${\bf u}(D) = {\sum}_{i} u_i D^{i}$, where $u_i \in {\mathbb F}_{q}$
and $u_i = 0$ for $i\leq r $, for some $r \in \mathbb{Z}$, we define
the weight of ${\bf u}(D)$ as wt$({\bf u}(D)) =
{\sum}_{\mathbb{Z}}$wt$(u_i)$. A generator matrix $G(D)$ is called
\emph{catastrophic} if there exists a ${\bf u}{(D)}^{k}\in {\mathbb
F}_{q}{((D))}^{k}$ of infinite Hamming weight such that ${\bf
u}{(D)}^{k}G(D)$ has finite Hamming weight. The AQCC's constructed
in this paper have noncatastrophic generator matrices since the
corresponding classical convolutional codes constructed here have
basic (and reduced) generator matrices. The Euclidean inner product
of two $n$-tuples ${\bf u}(D) = {\sum}_i {\bf u}_i D^i$ and ${\bf
v}(D) = {\sum}_j {\bf u}_j D^j$ in ${\mathbb F}_q {[D]}^{n}$ is
defined as $\langle {\bf u}(D)\mid {\bf v}(D)\rangle = {\sum}_i {\bf
u}_i \cdot {\bf v}_i$. If $C$ is a convolutional code then the code
$C^{\perp }=\{ {\bf u}(D) \in {\mathbb F}_q {[D]}^{n}\mid \langle
{\bf u}(D)\mid {\bf v}(D)\rangle = 0$ for all ${\bf v}(D)\in C\}$
denotes its Euclidean dual.

Let $C \subseteq { \mathbb F}_{q}^{n}$ an ${[n, k, d]}_{q}$ block
code with parity check matrix $H$. We split $H$ into $\mu+1$
disjoint submatrices $H_i$ such that $H = \left[
\begin{array}{c}
H_0\\
H_1\\
\vdots\\
H_{\mu}\\
\end{array}
\right],$ where each $H_i$ has $n$ columns, obtaining the polynomial
matrix $G(D) =  {\tilde H}_0 + {\tilde H}_1 D + {\tilde H}_2 D^2 +
\ldots + {\tilde H}_{\mu} D^{\mu}$. The matrices ${\tilde H}_i$,
$1\leq i\leq \mu$, are derived from the respective matrices $H_i$ by
adding zero-rows at the bottom such that ${\tilde H}_i$ has $\kappa$
rows in total, where $\kappa$ is the maximal number of rows among
the matrices $H_i$. The matrix $G(D)$ generates a convolutional code
$V$.

\begin{theorem}\cite[Theorem 3]{Aly:2007}\label{A}
Let $C \subseteq { \mathbb F}_{q}^{n}$ be a linear code with
parameters ${[n, k, d]}_{q}$. Assume that $H \in { \mathbb
F}_{q}^{(n-k)\times n}$ is a parity check matrix for $C$ partitioned
into submatrices $H_0, H_1, \ldots, H_{\mu}$ as above such that
$\kappa = \operatorname{rk}H_0$ and $\operatorname{rk}H_i \leq
\kappa$ for $1 \leq i\leq \mu$.\\
(a) The matrix $G(D)$ is a reduced basic generator matrix;\\
(b) If $d_f$ and $d_{f}^{\perp}$ denote the free distances of $V$
and $V^{\perp}$, respectively, $d_i$ denote the minimum distance of
the code $C_i = \{ {\bf v}\in { \mathbb F}_{q}^{n} \mid {\bf v}
{\tilde H}_i^t =0 \}$ and $d^{\perp}$ is the minimum distance of
$C^{\perp}$, then one has $\min \{ d_0 + d_{\mu} , d \} \leq
d_f^{\perp} \leq  d$ and $d_f \geq d^{\perp}$.
\end{theorem}


\subsection{Review of quantum convolutional codes}\label{IIB}

In this subsection, we recall the concept of quantum convolutional
code (QCC). For more details, the reader can consult
\cite{Aly:2007,Klapp:2007,Forney:2007}.

A quantum convolutional code is defined by means of its stabilizer
which is a subgroup of the infinite version of the Pauli group,
consisting of tensor products of generalized Pauli matrices acting
on a semi-infinite stream of qudits. The stabilizer can be defined
by a stabilizer matrix of the form $S(D) = ( X(D)\mid Z(D)) \in
{\mathbb F}_{q}{[D]}^{(n-k)\times 2n}$ satisfying $X(D){Z(1/D)}^{t}
- Z(D){X(1/D)}^{t}=0$ (symplectic orthogonality). Let ${\mathcal Q}$
be a QCC defined by a full-rank stabilizer matrix $S(D)$ given
above. The constraint length is defined as ${\gamma}_{i} =
{\max}_{1\leq j\leq n} \{ \max \{\deg X_{ij}(D),$ $\deg Z_{ij}(D)\}
\},$ and the overall constraint length as $\gamma
=\displaystyle\sum_{i=1}^{n-k} {\gamma}_{i}$. If $\gamma$ has the
smallest value among all basic generator matrices then $\gamma$ is
the \emph{degree} of the code. The memory $\mu$ of ${\mathcal Q}$ is
defined as $\mu = {\max}_{1\leq i\leq n-k, 1\leq j\leq n} \{ \max \{
\deg {X}_{ij}(D),$ $\deg {Z}_{ij} (D)\} \}$.

Here we define the free distance of a quantum convolutional code
\cite{Klapp:2007}. Let ${\mathbb H} = {\mathbb C}^{q^n} = {\mathbb
C}^{q} \otimes \ldots \otimes {\mathbb C}^{q}$ be the Hilbert space
and $\mid$$x \rangle$ be the vectors of an orthonormal basis of
${\mathbb C}^{q}$, where $ x \in {\mathbb F}_{q}$. Let $a, b \in
{\mathbb F}_{q}$ and take the unitary operators $X(a)$ and $Z(b)$ in
${\mathbb C}^{q}$ defined by $X(a)$$\mid$$x \rangle =$$\mid$$x +
a\rangle$ and $Z(b)$$\mid$$x \rangle = w^{tr(bx)}$$\mid$$x\rangle$,
respectively, where $w=\exp (2\pi i/ p)$ is a primitive $p$-th root
of unity, $p$ is the characteristic of ${\mathbb F}_{q}$ and
$\operatorname{tr}$ is the trace map from ${\mathbb F}_{q}$ to
${\mathbb F}_{p}$. Let ${\mathbb E} = \{X(a), Z(b) | a, b \in
{\mathbb F}_{q} \}$ be the \emph{error basis}. The set $P_{\infty}$
(according to \cite{Klapp:2007}) is the set of all infinite tensor
products of matrices $N\in \langle M\mid M \in {\mathbb E} \rangle$,
in which all but finitely many tensor components are equal to $I$,
where $I$ is the $q\times q$ identity matrix. Then one defines the
\emph{weight} wt of $A\in P_{\infty}$ as its (finite) number of
nonidentity tensor components. In this context, one says that a
quantum convolutional code has free distance $d_{f}$ if and only if
it can detect all errors of weight less than $d_{f}$, but cannot
detect some error of weight $d_{f}$. Then ${\mathcal Q}$ is a rate
$k/n$ code with parameters $[(n, k, \mu; \gamma, d_{f} ){]}_{q}$,
where $n$ is the frame size, $k$ is the number of logical qudits per
frame, $\mu$ is the memory, $\gamma$ is the degree and $d_{f}$ is
the free distance of the code.

On the other hand, a quantum convolutional code can also be
described in terms of a semi-infinite stabilizer matrix $S$ with
entries in ${\mathbb F}_{q} \times {\mathbb F}_{q}$ in the following
way. If $S(D)=\displaystyle\sum_{i=0}^{\mu}G_{i}D^{i}$, where each
matrix $G_{i}$ for all $i=0, \ldots, \mu$, is a matrix of size $(n -
k) \times n$, then the semi-infinite matrix is defined as
\begin{eqnarray*}
S = \left[
\begin{array}{cccccccc}
G_0 & G_1 & \ldots & G_{\mu} & 0 & \ldots & \ldots & \ldots\\
0 & G_0 & G_1 & \ldots & G_{\mu} & 0 & \ldots & \ldots\\
0 & 0 & G_0 & G_1 & \ldots & G_{\mu} & 0 & \ldots\\
\vdots & \vdots & \vdots & \vdots & \vdots & \vdots & \vdots & \vdots\\
\end{array}
\right].
\end{eqnarray*}
Let us recall the well known CSS-like construction:

\begin{theorem}\cite{Steane:1996,Calderbank:1998,Ketkar:2006}
(CSS-like Construction) Let $C_1$ and $C_2$ be two classical
convolutional codes with parameters $(n, k_1)_{q}$ and $(n,
n-k_2)_{q}$, respectively, such that $C_2^{\perp}\subset C_1$. The
stabilizer matrix is given by
\begin{eqnarray*}
\left(
\begin{array}{ccc}
H_2 (D) & | & 0\\
0 & | & H_1 (D)\\
\end{array}
\right) \in {\mathbb F}_{q}[D]^{(n-k_1 +k_2)\times 2n},
\end{eqnarray*}
where $H_1 (D)$ and $H_2 (D)$ denote parity check matrices of $C_1$
and $C_2$, respectively. Then there exists an $[(n, K=k_1 - k_2,
{(d_{z})}_{f}/{(d_{x})}_{f})]_{q}$ convolutional stabilizer code,
where ${(d_{x})}_{f} = \min \{ \operatorname{wt}(C_1 \backslash
C_2^{\perp}), \operatorname{wt}( C_2 \backslash C_1^{\perp}) \}$ and
${(d_{z})}_{f} = \max \{ \operatorname{wt}(C_1 \backslash
C_2^{\perp}),$ $\operatorname{wt}( C_2 \backslash C_1^{\perp}) \}$.
\end{theorem}
\begin{remark}
To avoid overly burdensome notation, we assume throughout this paper
that if ${(d_{x})}_{f}> {(d_{z})}_{f}$ then the values are changed.
\end{remark}


\section{Asymmetric quantum convolutional codes}\label{III}

In this section we present the contributions of this paper. As it
was said previously, we construct the first families of AQCC's by
means of algebraic methods. More specifically, we construct reduced
basic generator matrices for two classical convolutional codes $V_1$
and $V_2$, where $V_{2} \subset V_1$, in order to apply the CSS-type
construction. This section is divided in three subsections, which
contain three distinct code constructions.

\subsection{Construction I}\label{subIIIA}

In this section we present the first construction method of this
paper. Theorem~\ref{main1} establishes the existence of AQQC�s:

\begin{theorem}\label{main1}(General Construction)
Let $q$ be a prime power and $n$ be a positive integer. Then there
exist asymmetric quantum convolutional codes with parameters
$$[(n, \operatorname{rk}H_0, {\mu}^{*}; {\gamma}_{1}+{\gamma}_{2},
{(d_{z})}_{f}/{(d_{x})}_{f})]_{q},$$ where ${\gamma}_{1}=
\mu(\operatorname{rk}H_{\mu}+\operatorname{rk}H_{\mu}^{'})+\displaystyle\sum_{i=1}^{\mu-1}(\mu
-i)[\operatorname{rk}H_{(\mu-i)}^{'}-\operatorname{rk}H_{(\mu-i+1)}^{'}]$,
${\gamma}_{2}=
\mu(\operatorname{rk}H_{\mu}^{'})+\displaystyle\sum_{i=1}^{\mu-1}(\mu
-i)[\operatorname{rk}H_{(\mu-i)}^{'}-\operatorname{rk}H_{(\mu-i+1)}^{'}]$,
${(d_{x})}_{f}\geq{(d_1)}_{f} \geq d^{\perp}$ and ${(d_{z})}_{f}\geq
{(d_{2})}_f^{\perp} $, where ${(d_1)}_{f}, d^{\perp},
{(d_{2})}_f^{\perp}$ and the matrices $H_{0}, H_{0}^{'}, H_1,
H_{1}^{'}, \ldots, H_{\mu}, H_{\mu}^{'}$, are constructed below.
\end{theorem}
\begin{proof}
Consider a set of $m < n$ linearly independent (LI) vectors ${\bf
v}_{i}\in { \mathbb F}_{q}^{n}$, $i=1, 2, \ldots, m$. Let
\begin{eqnarray*}
{\mathcal H}=\left[\begin{array}{c}
H_0\\
H_{0}^{'}\\
H_1\\
H_{1}^{'}\\
\vdots\\
H_{\mu}\\
H_{\mu}^{'}\\
\end{array}
\right]
\end{eqnarray*}
be the matrix whose rows are the vectors ${\bf v}_{i}$, $i=1, 2,
\ldots, m$. The matrices $H_{0}, H_{0}^{'}, H_1, H_{1}^{'}, \ldots,
H_{\mu}, H_{\mu}^{'}$, are mutually disjoint. The matrices $H_i$,
$i=0, 1, \ldots, \mu$, are chosen in such a way that
$\operatorname{rk}H_{i}= \operatorname{rk}H_{j}$, for all $i,j=0,1,
\ldots, \mu$ (the choice of the vectors in each $H_i$ is arbitrary).
In order to compute the degree of the convolutional code constructed
in the sequence, we assume that $H_{0}^{'}$ has full rank and $
\operatorname{rk}H_{0}^{'}\geq
 \operatorname{rk}H_{1}^{'}\geq \ldots \geq \operatorname{rk} H_{\mu}^{'}$.
The matrices ${{\tilde H}}_{i}^{'}$ with $1\leq i\leq \mu$, are
obtained from the respective matrices $H_i^{'}$ by adding zero-rows
at the bottom such that ${\tilde H}_i$ has
$\operatorname{rk}H_{0}^{'}$ rows in total.

Let ${\mathcal H}$ be a parity check matrix of a linear block code
$C=[n, k, d{]}_{q}$, where $k=n - m$. Consider the linear block code
$C^{*}=[n, k^{*}, d^{*}]_{q}$ with parity check matrix
\begin{eqnarray*}
{H}^{*}=\left[\begin{array}{c}
H_{0}^{'}\\
H_{1}^{'}\\
\vdots\\
H_{\mu}^{'}\\
\end{array}
\right].
\end{eqnarray*}
Next, we construct a matrix $G_1(D)$ as follows:
\begin{eqnarray*}
G_1(D) =
 \left[\begin{array}{c}
H_0\\
--\\
{H}_{0}^{'}\\
\end{array}
\right]+ \left[\begin{array}{c}
H_1\\
--\\
{{\tilde H}}_{1}^{'}\\
\end{array}
\right]D + \left[\begin{array}{c}
H_2\\
--\\
{{\tilde H}}_{2}^{'}\\
\end{array}
\right]D^{2} + \ldots + \left[\begin{array}{c}
{H}_{\mu}\\
--\\
{\tilde H}_{\mu}^{'}\\
\end{array}
\right]D^{\mu}.
\end{eqnarray*}

Further, let us consider the submatrices $G_{0}(D)$ and $G_{2}(D)$
of $G_{1}(D)$, given, respectively, by
$$G_{0}(D) =  H_0 + {H}_1 D + {H}_2 D^{2} +
\ldots + {H}_{\mu} D^{\mu}$$ and
$$G_{2}(D) =  {H}_{0}^{'} + {\tilde H}_{1}^{'} D + {\tilde
H}_{2}^{'} D^2 + \ldots + {\tilde H}_{\mu}^{'} D^{\mu}.$$ We know
that $G_{1}(D) \in {\mathbb F}_{q}{[D]}^{\kappa \times n}$,
\emph{i.e.}, $G_{1}(D)$ has full rank
$\kappa=\operatorname{rk}H_{0}+\operatorname{rk}H_{0}^{'}$;
$G_{2}(D)$ has full rank $k_{0}^{'}=\operatorname{rk}H_{0}^{'}$.
From construction, it follows that $G_{1}(D)$ and $G_{2}(D)$ are
reduced basic generator matrices of convolutional codes $V_1$ and
$V_2$, respectively. Both convolutional codes have memory $\mu$.
Applying a similar idea as in the proof of \cite[Theorem
3]{Aly:2007}, the free distance ${(d_1)}_{f}$ of the convolutional
code $V_1$ and the free distance ${(d_1)}_{f}^{\perp}$ of its
Euclidean dual $V_{1}^{\perp}$ satisfy $\min \{ D_0 + D_{\mu} , d \}
\leq {(d_1)}_{f}^{\perp} \leq d$ and ${(d_1)}_{f} \geq d^{\perp}$,
where $D_0$ is the minimum distance of the code with parity check
matrix $\left[\begin{array}{c}
H_0\\
--\\
{H}_{0}^{'}\\
\end{array}
\right]$ and $D_{\mu}$ is the minimum distance of the code with
parity check matrix $\left[\begin{array}{c}
{H}_{\mu}\\
--\\
{{\tilde H}}_{\mu}^{'}\\
\end{array}
\right]$. Similarly, the free distance ${(d_{2})}_f$ of $V_2$ and
the free distance ${(d_{2})}_f^{\perp}$ of $V_{2}^{\perp}$ satisfy
$\min \{ d_0^{'} + d_{\mu}^{'} , d^{*} \} \leq {(d_{2})}_f^{\perp}
\leq d^{*}$ and ${(d_{2})}_f \geq {(d^{\perp})}^{*}$, where
$d_0^{'}$ is the minimum distance of the code $C_{0}^{'}$ with
parity check matrix ${H}_{0}^{'}$ and $d_{\mu}^{'}$ is the minimum
distance of the code with parity check matrix ${{\tilde
H}}_{\mu}^{'}$. The degree ${\gamma}_{2}$ of $V_2$ equals
${\gamma}_{2}=
\mu(\operatorname{rk}H_{\mu}^{'})+\displaystyle\sum_{i=1}^{\mu-1}(\mu
-i)[\operatorname{rk}H_{(\mu-i)}^{'}-\operatorname{rk}H_{(\mu-i+1)}^{'}]$;
the code $V_{2}^{\perp}$ also has degree ${\gamma}_{2}$. On the
other hand, the degree ${\gamma}_{1}$ of $V_1$ is equal to
${\gamma}_{1}=
\mu(\operatorname{rk}H_{\mu}+\operatorname{rk}H_{\mu}^{'})+\displaystyle\sum_{i=1}^{\mu-1}(\mu
-i)[\operatorname{rk}H_{(\mu-i)}^{'}-\operatorname{rk}H_{(\mu-i+1)}^{'}]$;
$V_{1}^{\perp}$ also has degree ${\gamma}_{1}$.

We know that $V_2 \subset V_1$. The corresponding CSS-type code
derived from $V_1$ and $V_{2}$ has frame size $n$,
$k=\operatorname{rk}H_0$ logical qudits per frame, degree $\gamma =
{\gamma}_{1}+{\gamma}_{2}$, ${(d_{x})}_{f}\geq{(d_1)}_{f} \geq
d^{\perp}$ and ${(d_{z})}_{f}\geq {(d_{2})}_f^{\perp} $, where $\min
\{ d_0^{'} + d_{\mu}^{'} , d^{*} \} \leq {(d_{2})}_f^{\perp} \leq
d^{*}$. Thus one can get an $[(n, \operatorname{rk}H_0, {\mu}^{*};
{\gamma}_{1}+{\gamma}_{2}, {(d_{z})}_{f}/{(d_{x})}_{f})]_{q}$ AQCC.
If $H_1(D)$ is a generator matrix of the code $V_{1}^{\perp}$ then a
stabilizer matrix of our AQCC is given by
\begin{eqnarray*}
\left(
\begin{array}{ccc}
G_2(D) & | & 0\\
0 & | & H_1(D)\\
\end{array}
\right).
\end{eqnarray*}
Other variant of this construction can be obtained by considering a
CSS-type code derived from the pair of classical convolutional codes
$V_{1}^{\perp} \subset V_{2}^{\perp}$. The proof is complete.
\end{proof}

\subsection{Construction II}\label{subIIIB}

Let $q$ be a prime power and $n$ a positive integer such that $\gcd
(q, n) = 1$. Let $\alpha$ be a primitive $n$-th root of unity in
some extension field. Recall that a cyclic code $C$ of length $n$
over ${\mathbb F}_{q}$ is a Bose-Chaudhuri-Hocquenghem (BCH) code
with designed distance $\delta$ if, for some integer $b\geq 0$, we
have $g(x)= \operatorname{l.c.m.} \{{M}^{(b)}(x),
{M}^{(b+1)}(x),\ldots, {M}^{(b+\delta-2)}(x)\}$, \emph{i.e.}, $g(x)$
is the monic polynomial of smallest degree over $F_{q}$ having
${{\alpha}^{b}}, {{\alpha}^{b+1}},$ $\ldots,
{{\alpha}^{b+\delta-2}}$ as zeros. Therefore, $c\in C$ if and only
if $c({\alpha}^{b})=c({{\alpha}^{b+1}})=\ldots =
c({{\alpha}^{b+\delta-2}})=0$. Thus the code has a string of $\delta
- 1$ consecutive powers of $\alpha$ as zeros. It is well known that
the minimum distance of a BCH code is greater than or equal to its
designed distance $\delta$. A parity check matrix for $C$ is given
by
\begin{eqnarray*}
H_{\delta , b} =  \left[
\begin{array}{ccccc}
1 & {{\alpha}^{b}} & {{\alpha}^{2b}} & \cdots & {{\alpha}^{(n-1)b}} \\
1 & {{\alpha}^{(b+1)}} & {{\alpha}^{2(b+1)}} & \cdots & {{\alpha}^{(n-1)(b+1)}}\\
\vdots & \vdots & \vdots & \vdots & \vdots\\
1 & {{\alpha}^{(b+\delta-2)}} & \cdots & \cdots & {{\alpha}^{(n-1)(b+\delta-2)}}\\
\end{array}
\right],
\end{eqnarray*}
where each entry is replaced by the corresponding column of $l$
elements from ${\mathbb F}_{q}$, where $l=
{\operatorname{ord}}_{n}(q)$, and then removing any linearly
dependent rows. The rows of the resulting matrix over ${\mathbb
F}_{q}$ are the parity checks satisfied by $C$.

Let us recall a useful results shown in \cite{LaGuardia:2014}:

\begin{theorem}\cite[Theorem 4.2]{LaGuardia:2014}\label{laguardiamds}
Assume that $q=2^{t}$, where $t\geq 3$ is an integer, $n=q+1$ and
consider that $a=\frac{q}{2}$. Then there exist classical MDS
convolutional codes with parameters $(n, n-2i, 2; 1, 2i+3)_{q}$,
where $1\leq i \leq a - 1$.
\end{theorem}

Theorem~\ref{main2} establishes conditions in which it is possible
to construct AQCC's derived from BCH codes.

\begin{theorem}\label{main2}
Let $q=2^{t}$, where $t\geq 4$ and consider that $n=q+1$ and
$a=\frac{q}{2}$. Then there exists an AQCC with parameters $[(n,
2i-4, {\mu}^{*}; 6, {[d_{z}]}_{f}/{[d_{x}]}_{f})]_{q}$, where
${(d_{z})}_{f}\geq n-2i-1$ and ${(d_{x})}_{f}\geq 3$, for all $3\leq
i \leq a - 1$.
\end{theorem}

\begin{proof}
Consider the parity check ${\mathbb F}_{q}$-matrix of the BCH code
$C$ given by
\begin{eqnarray*}
{\mathcal H} = \left[
\begin{array}{ccccc}
1 & {{\alpha}^{a}} & \cdots & \cdots & {{\alpha}^{(n-1)a}}\\
1 & {{\alpha}^{(a-1)}} & \cdots & \cdots & {{\alpha}^{(n-1)(a-1)}}\\
\vdots & \vdots & \vdots & \vdots & \vdots\\
1 & {{\alpha}^{(a-i+1)}} & {{\alpha}^{2(a-i+1)}} & \cdots & {{\alpha}^{(n-1)(a-i+1)}}\\
1 & {{\alpha}^{(a-i)}} & {{\alpha}^{2(a-i)}} & \cdots & {{ \alpha}^{(n-1)(a-i)}} \\
\end{array}
\right],
\end{eqnarray*}
whose entries are expanded with respect to some $ {\mathbb
F}_{q}$-basis ${\mathcal B}$ of ${\mathbb F}_{q^2}$, after removing
the linearly dependent rows. This BCH code was constructed in the
proof of \cite[Theorem 4.2]{LaGuardia:2014} (more precisely, it is
the code $C_2$ constructed there); $C$ is a MDS code with parameters
${[n, n-2i-2, 2i+3]}_{q}$. Its (Euclidean) dual code $C^{\perp}$ is
also a MDS code with parameters ${[n, 2i+2, n-2i-1]}_{q}$.

Next, we construct a classical convolutional code $V_{1}$ generated
by the reduced basic matrices
\begin{eqnarray*}
G_{1}(D) = \left[
\begin{array}{ccccc}
1 & {{\alpha}^{(a-i+2)}} & {{\alpha}^{2(a-i+2)}} & \cdots & {{
\alpha}^{(n-1)(a-i+2)}}\\
- & - & - & - & -\\
1 & {{\alpha}^{a}} & \cdots & \cdots & {{\alpha}^{(n-1)a}}\\
1 & {{\alpha}^{(a-1)}} & \cdots & \cdots & {{\alpha}^{(n-1)(a-1)}}\\
\vdots & \vdots & \vdots & \vdots & \vdots\\
1 & {{\alpha}^{(a-i+3)}} & {{\alpha}^{2(a-i+3)}} & \cdots & {{\alpha}^{(n-1)(a-i+3)}}\\
\end{array}
\right]+\\ \left[
\begin{array}{ccccc}
1 & {{\alpha}^{(a-i+1)}} & {{\alpha}^{2(a-i+1)}} & \cdots & {{\alpha}^{(n-1)(a-i+1)}}\\
- & - & - & - & -\\
1 & {{\alpha}^{(a-i)}} & {{\alpha}^{2(a-i)}} & \cdots & {{\alpha}^{(n-1)(a-i)}}\\
0 & 0 & 0 & 0 & 0\\
\vdots & \vdots & \vdots & \vdots & \vdots\\
0 & 0 & 0 & 0 & 0\\
\end{array}
\right]D
\end{eqnarray*}
and
\begin{eqnarray*}
G_{2}(D) = \left[
\begin{array}{ccccc}
1 & {{\alpha}^{(a-i+2)}} & {{\alpha}^{2(a-i+2)}} & \cdots & {{
\alpha}^{(n-1)(a-i+2)}}\\
\end{array}
\right]+\\ \left[
\begin{array}{ccccc}
1 & {{\alpha}^{(a-i+1)}} & {{\alpha}^{2(a-i+1)}} & \cdots & {{\alpha}^{(n-1)(a-i+1)}}\\
\end{array}
\right]D
\end{eqnarray*}
The code $V_1$, generated by $G_1(D)$, is a unit memory code of
dimension $k_1 = 2(i-1)$ and degree ${\gamma}_{1}=4$; $V_1$ is an
$(n, 2[i-1], 4; 1, {[d_{1}]}_{f}\geq n-2i-1)_{q}$ code. Its
Euclidean dual code $V_1^{\perp}$ has parameters $(n, n-2[i-1], 4;
{\mu}_{1}^{\perp}, {[d_{1}]}_{f}^{\perp}\geq 2i+2)_{q}$. The code
$V_{2}$, generated by $G_{2}(D)$, is an $(n, 2, 2; 1,
{[d_{2}]}_{f})_{q}$ code, so $V_{2}^{\perp}$ has parameters $(n,
n-2, 2; {\mu}_{2}^{\perp}, {[d_{2}]}_{f}^{\perp}\geq 3)_{q}$. From
construction, it follows that $V_{2}\subset V_1$, so
$V_{1}^{\perp}\subset V_{2}^{\perp}$. Consider the stabilizer matrix
given by
\begin{eqnarray*}
\left(
\begin{array}{ccc}
H_1(D) & | & 0\\
0 & | & G_{2}(D)\\
\end{array}
\right),
\end{eqnarray*}
where $H_1(D)$ is a parity check matrix of the code $V_{1}^{\perp}$.
The corresponding CSS-type code has $K=2i-4$, $\gamma = 6$,
${(d_{z})}_{f}\geq n-2i-1$ and ${(d_{x})}_{f}\geq 3$. Thus there
exists an $[(n, 2i-4, {\mu}^{*}; 6,
{[d_{z}]}_{f}/{[d_{x}]}_{f})]_{q}$ AQCC.
\end{proof}

\begin{remark}
It is interesting to note that the idea of construction of the
matrix $G_{2}(D)$ shown in the proof of Theorem~\ref{main2} is
distinct from that given in Theorem~\ref{main1}.
\end{remark}


\begin{theorem}\label{main3}
Let $q=2^{l}$, where $l\geq 4$ and consider that $n=q+1$ and
$a=\frac{q}{2}$. Then there exist AQCC's with parameters
\begin{itemize}
\item [ a)] $[(n, 2i-2t-2, {\mu}^{*}; 6,
{[d_{z}]}_{f}/{[d_{x}]}_{f})]_{q}$, where ${(d_{z})}_{f}\geq
n-2i-1$, ${(d_{x})}_{f}\geq 2t+3$, $i$ and $t$ are positive integers
such that $1\leq t\leq i-2$ and $3\leq i \leq a - 1$;

\item [ b)] $[(n, 2i-2t, {\mu}^{*}; 4, {[d_{z}]}_{f}/{[d_{x}]}_{f})]_{q}$, where
${(d_{z})}_{f}\geq n-2i-1$, ${(d_{x})}_{f}\geq 2t+3$, $i$ and $t$
are positive integers such that $1\leq t\leq i-1$ and $2\leq i \leq
a - 1$.
\end{itemize}
\end{theorem}

\begin{proof}
We only show Item a), since Item b) is similar. The notation and the
matrix ${\mathcal H}$ is the same as in the proof of
Theorem~\ref{main2}. We split ${\mathcal H}$ into disjoint
submatrices in order to construct a reduced basic generator matrix
$G_1(D)$ of the code $V_1$, given by
\begin{eqnarray*}
G_{1}(D) = \left[
\begin{array}{ccccc}
1 & {{\alpha}^{[a-(t+1)]}} & {{\alpha}^{2[a-(t+1)]}} & \cdots & {{
\alpha}^{(n-1)[a-(t+1)]}}\\
1 & {{\alpha}^{a}} & \cdots & \cdots & {{\alpha}^{(n-1)a}}\\
1 & {{\alpha}^{(a-1)}} & \cdots & \cdots & {{\alpha}^{(n-1)(a-1)}}\\
\vdots & \vdots & \vdots & \vdots & \vdots\\
1 & {{\alpha}^{[a-(t-1)]}} & {{\alpha}^{2[a-(t-1)]}} & \cdots & {{\alpha}^{(n-1)[a-(t-1)]}}\\
- & - & - & - & -\\
1 & {{\alpha}^{[a-(t+2)]}} & {{\alpha}^{2[a-(t+2)]}} & \cdots & {{
\alpha}^{(n-1)[a-(t+2)]}}\\
\vdots & \vdots & \vdots & \vdots & \vdots\\
1 & {{\alpha}^{[a-(i-2)]}} & {{\alpha}^{2[a-(i-2)]}} & \cdots &
{{\alpha}^{(n-1)[a-(i-2)]}}\\
1 & {{\alpha}^{[a-(i-1)]}} & {{\alpha}^{2[a-(i-1)]}} & \cdots &
{{\alpha}^{(n-1)[a-(i-1)]}}\\
\end{array}
\right]+\\ \left[
\begin{array}{ccccc}
1 & {{\alpha}^{(a-i)}} & {{\alpha}^{2(a-i)}} & \cdots & {{\alpha}^{(n-1)(a-i)}}\\
1 & {{\alpha}^{(a-t)}} & {{\alpha}^{2(a-t)}} & \cdots & {{
\alpha}^{(n-1)(a-t)}}\\
0 & 0 & 0 & 0 & 0\\
\vdots & \vdots & \vdots & \vdots & \vdots\\
0 & 0 & 0 & 0 & 0\\
- & - & - & - & -\\
0 & 0 & 0 & 0 & 0\\
\vdots & \vdots & \vdots & \vdots & \vdots\\
0 & 0 & 0 & 0 & 0\\
\end{array}
\right]D,
\end{eqnarray*}
Let $V_{2}$ be the convolutional code generated by the reduced basic
matrix $G_{2}(D)$
\begin{eqnarray*}
G_{2}(D) = \left[
\begin{array}{ccccc}
1 & {{\alpha}^{a}} & \cdots & \cdots & {{\alpha}^{(n-1)a}}\\
1 & {{\alpha}^{(a-1)}} & \cdots & \cdots & {{\alpha}^{(n-1)(a-1)}}\\
\vdots & \vdots & \vdots & \vdots & \vdots\\
1 & {{\alpha}^{[a-(t-1)]}} & {{\alpha}^{2[a-(t-1)]}} & \cdots & {{\alpha}^{(n-1)[a-(t-1)]}}\\
\end{array}
\right]+\\ \left[
\begin{array}{ccccc}
1 & {{\alpha}^{(a-t)}} & {{\alpha}^{2(a-t)}} & \cdots & {{
\alpha}^{(n-1)(a-t)}}\\
0 & 0 & 0 & 0 & 0\\
\vdots & \vdots & \vdots & \vdots & \vdots\\
0 & 0 & 0 & 0 & 0\\
\end{array}
\right]D
\end{eqnarray*}
It is easy to see that the code $V_1$ has parameters $(n, 2i-2, 4;
1, {[d_{1}]}_{f}\geq n-2i-1)_{q}$ and $V_1^{\perp}$ has parameters
$(n, n-2i+2 , 4; {\mu}_{1}^{\perp}, {[d_{1}]}_{f}^{\perp})_{q}$. The
code $V_{2}$, generated by $G_{2}(D)$, is an $(n, 2t, 2; 1,
{[d_{2}]}_{f})_{q}$ code, so $V_{2}^{\perp}$ has parameters $(n,
n-2t, 2; {\mu}_{2}^{\perp}, {[d_{2}]}_{f}^{\perp}\geq 2t+3 )_{q}$.
Since $V_{2}\subset V_1$, it follows that $V_{1}^{\perp}\subset
V_{2}^{\perp}$. Thus there exists an $[(n, 2i-2t-2, {\mu}^{*}; 6,
{[d_{z}]}_{f}/{[d_{x}]}_{f})]_{q}$ AQCC, where ${(d_{z})}_{f}\geq
n-2i-1$ and ${(d_{x})}_{f}\geq 2t+3$.
\end{proof}

\begin{example}
Applying Theorem~\ref{main3}, one can get AQCC's with parameters\\
$[(17, 6, {\mu}^{*}; 6, {[d_{z}]}_{f}\geq 6/{[d_{x}]}_{f}\geq
5)]_{16}$, $[(17, 8, {\mu}^{*}; 4, {[d_{z}]}_{f}\geq
6/{[d_{x}]}_{f}\geq 5)]_{16}$, $[(17, 8, {\mu}^{*}; 6,
{[d_{z}]}_{f}\geq 5/{[d_{x}]}_{f}\geq 4)]_{16}$, $[(17, 10,
{\mu}^{*}; 4, {[d_{z}]}_{f}\geq 5/{[d_{x}]}_{f}\geq 4)]_{16}$,
$[(33, 24, {\mu}^{*}; 6, {[d_{z}]}_{f}\geq 5/{[d_{x}]}_{f}\geq
4)]_{32}$, $[(33, 26, {\mu}^{*}; 4, {[d_{z}]}_{f}\geq
5/{[d_{x}]}_{f}\geq 4)]_{32}$, $[(33, 22, {\mu}^{*}; 6,
{[d_{z}]}_{f}\geq 6/{[d_{x}]}_{f}\geq 5)]_{32}$, $[(33, 24,
{\mu}^{*}; 4, {[d_{z}]}_{f}\geq 6/{[d_{x}]}_{f}\geq 5)]_{32}$,
$[(33, 20, {\mu}^{*}; 6, {[d_{z}]}_{f}\geq 8/{[d_{x}]}_{f}\geq
5)]_{32}$, $[(33, 22, {\mu}^{*}; 4, {[d_{z}]}_{f}\geq
8/{[d_{x}]}_{f}\geq 5)]_{32}$ and so on.
\end{example}

\begin{theorem}\label{main4}
Assume that $q=p^{l}$, where $p$ is an odd prime and $l\geq 2$.
Consider that $n=q+1$ and $a=\frac{n}{2}$. Then there exist AQCC's
with parameters
\begin{itemize}
\item [ a)] $[(n, 2i-2t-2, {\mu}^{*}; 6, {[d_{z}]}_{f}/{[d_{x}]}_{f})]_{q}$,
where ${(d_{z})}_{f}\geq n-2i$ and ${(d_{x})}_{f}\geq 2t+2$, for all
$1\leq t\leq i-2$, where $3\leq i \leq a - 1$;

\item [ b)] $[(n, 2i-2t, {\mu}^{*}; 4, {[d_{z}]}_{f}/{[d_{x}]}_{f})]_{q}$,
where ${(d_{z})}_{f}\geq n-2i$ and ${(d_{x})}_{f}\geq 2t+2$, for all
$1\leq t\leq i-1$, where $2\leq i \leq a - 1$.
\end{itemize}
\end{theorem}

\begin{proof}
Analogous to that of Theorem~\ref{main3}.
\end{proof}

\begin{remark}
One more time we call the attention that the idea of construction of
the matrix $G_{2}(D)$ is different for each of Theorems~\ref{main1},
\ref{main2}, \ref{main3} and \ref{main4}. This remark also holds for
the results shown in Subsection~\ref{subIIIC}.
\end{remark}

\subsection{Construction III}\label{subIIIC}

In this subsection we are interested in constructing AQCC's derived
from Reed-Solomon (RS) and generalized Reed-Solomon (GRS) codes. We
first deal with RS codes. Recall that a RS code over ${\mathbb
F}_{q}$ is a BCH code, of length $n=q-1$, with parameters $[n,
n-d+1, d]_{q}$, where $2\leq d\leq n$. A parity check matrix of a RS
code is given by
\begin{eqnarray*}
H_{\delta , b} =  \left[
\begin{array}{ccccc}
1 & {{\alpha}^{b}} & {{\alpha}^{2b}} & \cdots & {{\alpha}^{(n-1)b}} \\
1 & {{\alpha}^{(b+1)}} & {{\alpha}^{2(b+1)}} & \cdots & {{\alpha}^{(n-1)(b+1)}}\\
\vdots & \vdots & \vdots & \vdots & \vdots\\
1 & {{\alpha}^{(b+d-2)}} & \cdots & \cdots & {{\alpha}^{(n-1)(b+d-2)}}\\
\end{array}
\right],
\end{eqnarray*}
whose entries are in ${\mathbb F}_{q}$.

In Theorem~\ref{main5} presented in the following, we construct
AQCC's derived from RS codes:

\begin{theorem}\label{main5}
Assume that $q \geq 8$ is a prime power. Then there exist AQCC's
with parameters
\begin{itemize}
\item [ a)] $[(q-1, i-t-1, {\mu}^{*}; 3,
{[d_{z}]}_{f}/{[d_{x}]}_{f})]_{q}$, where ${(d_{z})}_{f}\geq q-i-1$,
${(d_{x})}_{f}\geq t+2$, for all $1\leq t\leq i-2$, where $3\leq i
\leq q-3$;

\item [ b)] $[(q-1, i-t, {\mu}^{*}; 2,
{[d_{z}]}_{f}/{[d_{x}]}_{f})]_{q}$, where ${(d_{z})}_{f}\geq q-i-1$,
${(d_{x})}_{f}\geq t+2$, for all $1\leq t\leq i-1$, where $2\leq i
\leq q-3$.
\end{itemize}
\end{theorem}
\begin{proof}
We only show Item a), since Item b) is similar. The construction is
the same as in the proof of Theorem~\ref{main3}, although the codes
have distinct parameters. More specifically, starting from a parity
check matrix ${\mathcal H}$
\begin{eqnarray*}
{\mathcal H} = \left[
\begin{array}{ccccc}
1 & {{\alpha}^{a}} & \cdots & \cdots & {{\alpha}^{(n-1)a}}\\
1 & {{\alpha}^{(a-1)}} & \cdots & \cdots & {{\alpha}^{(n-1)(a-1)}}\\
\vdots & \vdots & \vdots & \vdots & \vdots\\
1 & {{\alpha}^{(a-i+1)}} & {{\alpha}^{2(a-i+1)}} & \cdots & {{\alpha}^{(n-1)(a-i+1)}}\\
1 & {{\alpha}^{(a-i)}} & {{\alpha}^{2(a-i)}} & \cdots & {{ \alpha}^{(n-1)(a-i)}} \\
\end{array}
\right],
\end{eqnarray*}
of an $[q-1, q-i-2, i+2]_{q}$ RS code, we construct generator
matrices $G_1 (D)$ and $G_{2}(D)$ for codes $V_1$ and $V_{2}$,
respectively as per Theorem~\ref{main3}. In this context, it is easy
to see that $V_1$ is an $(q-1, i-1, 2; 1, {[d_{1}]}_{f}\geq
q-i-1)_{q}$ code, $V_1^{\perp}$ is an $(q-1, q-i, 2;
{\mu}_{1}^{\perp}, {[d_{1}]}_{f}^{\perp})_{q}$ code, $V_{2}$ is an
$(q-1, t, 1; 1, {[d_{2}]}_{f})_{q}$ and $V_{2}^{\perp}$ an $(q-1,
q-t-1, 1; {\mu}_{2}^{\perp}, {[d_{1}]}_{f}^{\perp}\geq t+2)_{q}$.
Then the corresponding CSS-type code has parameters $[(q-1, i-t-1,
{\mu}^{*}; 3, {[d_{z}]}_{f}/{[d_{x}]}_{f})]_{q}$, where
${(d_{z})}_{f}\geq q-i-1$ and ${(d_{x})}_{f}\geq t+2$.
\end{proof}

\begin{example}
By means of Theorem~\ref{main5}, one can construct AQCC's with
parameters $[(10, 4, {\mu}^{*}; 3, {[d_{z}]}_{f}\geq
4/{[d_{x}]}_{f}\geq 3)]_{11}$, $[(10, 5, {\mu}^{*}; 3,
{[d_{z}]}_{f}\geq 3/{[d_{x}]}_{f}\geq 3)]_{11}$, $[(10, 2,
{\mu}^{*}; 3, {[d_{z}]}_{f}\geq 6/{[d_{x}]}_{f}\geq 3 )]_{11}$,
$[(10, 1, {\mu}^{*}; 3, {[d_{z}]}_{f}\geq 6/{[d_{x}]}_{f}\geq 4
)]_{11}$, and so on.
\end{example}

Let us recall the definition of GRS codes. Let $n$ be an integer
such that $1\leq n\leq q$, and choose an $n$-tuple ${\bf
\zeta}=({\zeta}_{0}, \ldots, {\zeta}_{n-1})$ of distinct elements of
${\mathbb F}_{q}$. Assume that ${\bf v}=(v_{0}, \ldots, v_{n-1})$ is
an $n$-tuple of nonzero (not necessary distinct) elements of
${\mathbb F}_{q}$. For any integer $k$, $1\leq k\leq n$, consider
the set of polynomials of degree less than $k$, in ${\mathbb
F}_{q}[x]$, denoted by ${\mathcal P}_{k}$. Then we define the GRS
codes as ${\operatorname{GRS}}_{k}({\bf \zeta},{\bf
v})=\{(v_{0}f({\zeta}_{0}), v_{1}f({\zeta}_{1}), \ldots,
v_{n-1}f({\zeta}_{n-1}))| f \in{\mathcal P}_{k}\}$. It is well known
that ${\operatorname{GRS}}_{k}({\bf \zeta},{\bf v})$ is a MDS code
with parameters $[n, k, n - k + 1]_{q}$. The (Euclidean) dual
${\operatorname{GRS}}_{k}^{\perp}({\bf \zeta},{\bf v})$ of
${\operatorname{GRS}}_{k}({\bf \zeta},{\bf v})$ is also a GRS code
and ${\operatorname{GRS}}_{k}^{\perp}({\bf \zeta},{\bf
w})={\operatorname{GRS}}_{n-k}({\bf \zeta},{\bf v})$ for some
$n$-tuple ${\bf w}=(w_{0}, \ldots, w_{n-1})$ of nonzero elements of
${\mathbb F}_{q}$. A generator matrix of
${\operatorname{GRS}}_{k}({\bf \zeta},{\bf v})$ is given by
\begin{eqnarray*}
G = \left[
\begin{array}{cccc}
v_{0} & v_{1} & \cdots & v_{n-1}\\
v_{0}{\zeta}_{0} & v_{1}{\zeta}_{1} & \cdots & v_{n-1}{\zeta}_{n-1}\\
v_{0}{\zeta}_{0}^{2} & v_{1}{\zeta}_{1}^{2} & \cdots & v_{n-1}{\zeta}_{n-1}^{2}\\
\vdots & \vdots & \vdots & \vdots\\
v_{0}{\zeta}_{0}^{k-1} & v_{1}{\zeta}_{1}^{k-1} & \cdots & v_{n-1}{\zeta}_{n-1}^{k-1}\\
\end{array}
\right];
\end{eqnarray*}
a parity check matrix of ${\operatorname{GRS}}_{k}({\bf \zeta},{\bf
v})$ is
\begin{eqnarray*}
H = \left[
\begin{array}{cccc}
w_{0} & w_{1} & \cdots & w_{n-1}\\
w_{0}{\zeta}_{0} & w_{1}{\zeta}_{1} & \cdots & w_{n-1}{\zeta}_{n-1}\\
w_{0}{\zeta}_{0}^{2} & w_{1}{\zeta}_{1}^{2} & \cdots & w_{n-1}{\zeta}_{n-1}^{2}\\
\vdots & \vdots & \vdots & \vdots\\
w_{0}{\zeta}_{0}^{n-k-1} & w_{1}{\zeta}_{1}^{n-k-1} & \cdots & w_{n-1}{\zeta}_{n-1}^{n-k-1}\\
\end{array}
\right].
\end{eqnarray*}

In the next result, we construct new AQCC's derived from GRS codes.

\begin{theorem}\label{main6}
Let $q\geq 5$ be a prime power. Assume that $k \geq 1$ and $n \geq
5$ are integers such that $n\leq q$ and $k\leq n-4$. Choose an
$n$-tuple ${\bf \zeta}=({\zeta}_{0}, \ldots, {\zeta}_{n-1})$ of
distinct elements of ${\mathbb F}_{q}$ and an $n$-tuple ${\bf
v}=(v_{0}, \ldots, v_{n-1})$ of nonzero elements of ${\mathbb
F}_{q}$. Then there exists an $[(n, n-t-k-2, {\mu}^{*}; 3,
{[d_{z}]}_{f}/{[d_{x}]}_{f})]_{q}$ AQCC, where ${(d_{z})}_{f}\geq
t+2$ and ${(d_{x})}_{f}\geq k+1$, $1\leq t\leq n-k-2$.
\end{theorem}
\begin{proof}
Let
\begin{eqnarray*}
H = \left[
\begin{array}{cccc}
w_{0} & w_{1} & \cdots & w_{n-1}\\
w_{0}{\zeta}_{0} & w_{1}{\zeta}_{1} & \cdots & w_{n-1}{\zeta}_{n-1}\\
w_{0}{\zeta}_{0}^{2} & w_{1}{\zeta}_{1}^{2} & \cdots & w_{n-1}{\zeta}_{n-1}^{2}\\
\vdots & \vdots & \vdots & \vdots\\
w_{0}{\zeta}_{0}^{n-k-1} & w_{1}{\zeta}_{1}^{n-k-1} & \cdots & w_{n-1}{\zeta}_{n-1}^{n-k-1}\\
\end{array}
\right].
\end{eqnarray*}
be a parity check matrix of an ${\operatorname{GRS}}_{k}({\bf
\zeta},{\bf v})$ code. We split $H$ to form polynomial matrices $G_1
(D)$ and $G_{2}(D)$ of codes $V_1$ and $V_{2}$, respectively, as
follows:
\begin{eqnarray*}
G_{1}(D) = \left[
\begin{array}{cccc}
w_{0}{\zeta}_{0}^{n-k-3} & w_{1}{\zeta}_{1}^{n-k-3} & \cdots & w_{n-1}{\zeta}_{n-1}^{n-k-3}\\
w_{0} & w_{1} & \cdots & w_{n-1}\\
w_{0}{\zeta}_{0} & w_{1}{\zeta}_{1} & \cdots & w_{n-1}{\zeta}_{n-1}\\
\vdots & \vdots & \vdots & \vdots\\
w_{0}{\zeta}_{0}^{t-1} & w_{1}{\zeta}_{1}^{t-1} & \cdots & w_{n-1}{\zeta}_{n-1}^{t-1}\\
- & - & - & -\\
w_{0}{\zeta}_{0}^{t+1} & w_{1}{\zeta}_{1}^{t+1} & \cdots & w_{n-1}{\zeta}_{n-1}^{t+1}\\
\vdots & \vdots & \vdots & \vdots\\
w_{0}{\zeta}_{0}^{n-k-2} & w_{1}{\zeta}_{1}^{n-k-2} & \cdots & w_{n-1}{\zeta}_{n-1}^{n-k-1}\\
\end{array}
\right]+\\ \left[
\begin{array}{cccc}
w_{0}{\zeta}_{0}^{n-k-1} & w_{1}{\zeta}_{1}^{n-k-1} & \cdots & w_{n-1}{\zeta}_{n-1}^{n-k-1}\\
w_{0}{\zeta}_{0}^{t} & w_{1}{\zeta}_{1}^{t} & \cdots & w_{n-1}{\zeta}_{n-1}^{t}\\
0 & 0 & 0 & 0\\
\vdots & \vdots & \vdots & \vdots\\
0 & 0 & 0 & 0\\
- & - & - & -\\
0 & 0 & 0 & 0\\
\vdots & \vdots & \vdots & \vdots\\
0 & 0 & 0 & 0\\
\end{array}
\right]D
\end{eqnarray*}
and
\begin{eqnarray*}
G_{2}(D) = \left[
\begin{array}{cccc}
w_{0} & w_{1} & \cdots & w_{n-1}\\
w_{0}{\zeta}_{0} & w_{1}{\zeta}_{1} & \cdots & w_{n-1}{\zeta}_{n-1}\\
w_{0}{\zeta}_{0}^{2} & w_{1}{\zeta}_{1}^{2} & \cdots & w_{n-1}{\zeta}_{n-1}^{2}\\
\vdots & \vdots & \vdots & \vdots\\
w_{0}{\zeta}_{0}^{t-1} & w_{1}{\zeta}_{1}^{t-1} & \cdots & w_{n-1}{\zeta}_{n-1}^{t-1}\\
\end{array}
\right]+\\ \left[
\begin{array}{cccc}
w_{0}{\zeta}_{0}^{t} & w_{1}{\zeta}_{1}^{t} & \cdots & w_{n-1}{\zeta}_{n-1}^{t}\\
0 & 0 & 0 & 0\\
\vdots & \vdots & \vdots & \vdots\\
0 & 0 & 0 & 0\\
\end{array}
\right]D,
\end{eqnarray*}
where ${\bf w}=(w_{0}, \ldots, w_{n-1})$ is a vector such that
${\operatorname{GRS}}_{k}^{\perp}({\bf \zeta},{\bf
w})={\operatorname{GRS}}_{n-k}({\bf \zeta},{\bf v})$. The code $V_1$
has parameters $(n, n-k-2, 2; 1, {[d_{1}]}_{f}\geq k+1)_{q}$ and
$V_{1}^{\perp}$ has parameters $(n, k+2, 2; {\mu}_{1}^{\perp},
{[d_{1}]}_{f}^{\perp})_{q}$. Similarly, $V_{2}$ is an $(n, t, 1; 1,
{[d_{2}]}_{f})_{q}$ code and $V_{2}^{\perp}$ is an $(n, n-t, 1;
{\mu}_{2}^{\perp}, {[d_{1}]}_{f}^{\perp}\geq t+2)_{q}$ code. Then
there exists an $[(n, n-t-k-2, {\mu}^{*}; 3,
{[d_{z}]}_{f}/{[d_{x}]}_{f})]_{q}$ code, where ${(d_{z})}_{f}\geq
t+2$ and ${(d_{x})}_{f}\geq k+1$.
\end{proof}

\begin{example}
From Theorem~\ref{main6}, we can construct AQCC's with parameters
$[(5, 1, {\mu}^{*}; 3, {[d_{z}]}_{f}\geq 3/{[d_{x}]}_{f}\geq
2)]_{5}$, $[(7, 1, {\mu}^{*}; 3, {[d_{z}]}_{f}\geq
4/{[d_{x}]}_{f}\geq 3)]_{7}$, $[(8, 1, {\mu}^{*}; 3,$ $
{[d_{z}]}_{f}\geq 5/{[d_{x}]}_{f}\geq 3)]_{8}$, $[(17, 7, {\mu}^{*};
3, {[d_{z}]}_{f}\geq 7/{[d_{x}]}_{f}\geq 4)]_{17}$, $[(17, 7,
{\mu}^{*}; 3, {[d_{z}]}_{f}\geq 6/{[d_{x}]}_{f}\geq 5)]_{17}$,
$[(17, 6, {\mu}^{*}; 3, {[d_{z}]}_{f}\geq 7/{[d_{x}]}_{f}\geq
5)]_{17}$, $[(17, 4, {\mu}^{*}; 3, {[d_{z}]}_{f}\geq
9/{[d_{x}]}_{f}\geq 5)]_{17}$ and so on.
\end{example}

\begin{theorem}\label{main8}
Let $q\geq 5$ be a prime power. Assume that $k \geq 1$ and $n \geq
5$ are integers such that $n\leq q$ and $k\leq n-4$. Choose an
$n$-tuple ${\bf \zeta}=({\zeta}_{0}, \ldots, {\zeta}_{n-1})$ of
distinct elements of ${\mathbb F}_{q}$ and an $n$-tuple ${\bf
v}=(v_{0}, \ldots, v_{n-1})$ of nonzero elements of ${\mathbb
F}_{q}$. Then an $[(n, n-t-k-1, {\mu}^{*}; 2,
{[d_{z}]}_{f}/{[d_{x}]}_{f})]_{q}$ AQCC, where ${(d_{z})}_{f}\geq
t+2$, ${(d_{x})}_{f}\geq k+1$ and $1\leq t\leq n-k-1$ can be
constructed.
\end{theorem}
\begin{proof}
Similar to that of Theorem~\ref{main6}.
\end{proof}

\begin{example}
From Theorem~\ref{main8}, we obtain AQCC's with parameters $[(5, 1,
{\mu}^{*}; 2,$ ${[d_{z}]}_{f}\geq 4/{[d_{x}]}_{f}\geq 2)]_{5}$,
$[(7, 2, {\mu}^{*}; 2, {[d_{z}]}_{f}\geq 4/{[d_{x}]}_{f}\geq
3)]_{7}$, $[(7, 2, {\mu}^{*}; 2, {[d_{z}]}_{f}\geq
5/{[d_{x}]}_{f}\geq 2)]_{7}$, $[(7, 1, {\mu}^{*}; 2,
{[d_{z}]}_{f}\geq 5/{[d_{x}]}_{f}\geq 3)]_{7}$.
\end{example}


\section{Discussion}\label{IV}

Our main result is Theorem~\ref{main1}, which establishes a general
technique of construction for AQQC's. Subsection~\ref{subIIIB} is
concerned with constructions of AQQC's derived from classical
MDS-convolutional BCH codes and, in Subsection~\ref{subIIIC}, we
address the construction of AQQC's derived from classical
Reed-Solomon and generalized Reed-Solomon convolutional codes. It is
interesting to note that the choice of matrices $G_1(D)$ and
$G_{2}(D)$ was based on the fact that the corresponding (classical)
convolutional codes must be non-catastrophic, with great dimension
and minimum distances.

As was mentioned previously, this is the first work available in
literature dealing with constructions of asymmetric quantum
convolutional codes. Moreover, by applying algebraic techniques, we
have derived several families of such codes, and not only few codes
with specific parameters. However, much research remains to be done
in the area of AQCC's. In fact, there is no bound for the respective
free distances nor relationships among the parameters of AQCC's.
Other impossibility is to compare the parameters of the new AQCC's
with the ones displayed in literature, i.e., our codes have
parameters quite distinct of the QCC's available in literature. This
area of research needs much investigation, since it was introduced
recently (see \cite{Ollivier:2003}) in literature. Additionally,
even in the case of constructions of good QCC's, only few works are
displayed in literature
\cite{Forney:2007,Klapp:2007,LaGuardia:2014,LaGuardia:2014A}.
Moreover, the unique bound known in literature even for QCC's is the
generalized quantum Singleton bound (GQSB), introduced by
Klappenecker \emph{et al.} (see \cite{Klapp:2007}).

For future works, it will be interesting to establish analogous
results to (asymmetric) quantum generalized Singleton bound (see
\cite{Klapp:2007}), the (asymmetric) sphere packing bound among
other.

\section{Summary}\label{V}
We have constructed the first families of asymmetric quantum
convolutional codes available in literature. These new AQCC's are
derived from suitable families of classical convolutional codes with
good parameters, which have been also constructed in this paper. Our
codes have great asymmetry. Additionally, great variety of distinct
types of codes have also been presented. However, much work remains
to be done in order to find bounds for AQCC's as well as for the
development of such area of research.

\section*{Acknowledgment}
This research has been partially supported by the Brazilian Agencies
CAPES and CNPq.

%

\end{document}